\documentclass[reqno,11pt]{amsart}

\usepackage[utf8x]{inputenc}
\usepackage{amsmath, amsfonts, amssymb, amsthm, mathtools}
\usepackage{dsfont}
\usepackage{mathrsfs}
\usepackage{bm}

\usepackage{graphicx}
\DeclareGraphicsExtensions{.pdf,.png,.jpg}
\usepackage{rotating}
\usepackage{wasysym}
\usepackage{subfigure}

\usepackage{geometry}
\usepackage{setspace}
\usepackage{enumitem}
\setlist[enumerate]{leftmargin=2em}   

\usepackage[dvipsnames]{xcolor}
\usepackage[round,comma]{natbib}
\usepackage[colorlinks,urlcolor=red,citecolor=blue,linkcolor=teal]{hyperref}

\usepackage{comment}
\usepackage{nicefrac}
\usepackage{booktabs}
\usepackage{bbold}   
\usepackage{comment}

\newtheorem{theorem}{Theorem}
\newtheorem{corollary}[theorem]{Corollary}
\newtheorem{lemma}[theorem]{Lemma}

\theoremstyle{definition}
\newtheorem{example}[theorem]{Example}

\newtheorem{assumption}[theorem]{Assumption}

\numberwithin{equation}{section}

\DeclareMathOperator{\es}{\mathrm{ES}}

\newcommand{\R}{\mathbb{R}}

\newcommand{\E}{\mathbb{E}}
\renewcommand{\P}{\mathbb{P}}
\newcommand{\N}{\mathbb{N}}

\newcommand{\tn}{\textnormal}
\newcommand{\ind}{\mathbf{1}}
\newcommand{\Norm}{\|\cdot\|}

\newcommand{\dom}{\textnormal{dom}}

\newcommand{\sgn}{\textnormal{sgn}}
\newcommand{\ssd}{\leq_{\textnormal{icx}}}
\newcommand{\cx}{\le_{\tn{cx}}}

\newcommand{\CA}{\mathcal A}

\newcommand{\CF}{\mathcal F}
\newcommand{\CG}{\mathcal G}

\newcommand{\CL}{\mathcal L}

\newcommand{\CX}{\mathcal X}

\newcommand{\peq}{\preceq}

\newcommand{\mbf}{\mathbf}

\newcommand{\ph}{\varphi}

\begin{document}

\title{A unifying perspective on the collapse to the mean for law-invariant functionals}

    \author{Felix-Benedikt Liebrich}
		\address{Amsterdam School of Economics, University of Amsterdam, Netherlands.}
		\email{f.b.liebrich@uva.nl}
\date{July 24, 2026}

\begin{abstract}
We revisit the ``collapse to the mean'' phenomenon, which refers to mild structural conditions, such as local linearity, that force a law-invariant functional $\ph$ defined on finite-mean random variables to depend solely on the expectation of its argument $X$, and not on any other distributional feature.
Starting from a concise characterisation of the convex order, our simplified approach unifies and extends existing results without assuming the functional to be convex or monotone in the almost-sure order, and clarifies the conceptual foundations of the ``collapse to the mean" phenomenon.  
In addition, we establish a new ``dual collapse'' result for quasi-star-shaped functionals.

\smallskip

\noindent {\em JEL Classification:} C60 $\cdotp$ D81

\noindent {\em Mathematics Subject Classification (MSC2020):} 60E15
$\cdotp$ 91B05 $\cdotp$ 91B06

\noindent{\em Keywords:} Law invariance $\cdotp$ collapse to the mean $\cdotp$ convex order
\end{abstract}

\maketitle

\section{Introduction}

\thispagestyle{empty}

Defined on a domain $\CX$ of random variables $X$ over a fixed probability space, law-invariant functionals assign values solely on the basis of the distribution of the argument $X$ in question.
Such functionals are widely used in economics, finance, and risk management to assess the utility, risk, or deviation from a benchmark of a random prospect.

The simplest examples of law-invariant functionals are {\em expectation-based} ones, where the assigned value is entirely determined by the mean of the argument (assuming such a mean exists). 
Consequently, any two random variables with the same mean under the latent reference probability measure are mapped to the same value, irrespective of other distributional features in which they may differ, such as higher-order moments, symmetry, or tail behaviour.
By contrast, these features are often crucial for functionals assessing the riskiness of random prospects. 
They directly impact the computation of classical risk metrics such as Value-at-Risk, Expected Shortfall and upper partial moments \cite[Chapter 2.3]{Embrechts}. 
Insurance premia commonly add a deviation-measuring safety loading such as a fraction of variance or standard deviation to the actuarially fair price of a claim \cite[Chapter 5.3]{Kaas}. 
Moreover, since \cite{Markowitz}, it has been standard to select portfolios trading off the mean of random returns against their dispersion around the mean. Expectation-based functionals are therefore generally unsuited for risk management purposes.

In many contexts, however, such functionals are not pinned down in closed form directly but characterised by {\em axioms}, whose interplay and relation to other properties might not be transparent at first sight.
In this situation, the term ``collapse to the mean'' summarises at times surprising and seemingly unrelated conditions under which a law-invariant functional is necessarily---and potentially unintentionally---expectation-based.  
Following the discussion in \cite{Bellini}, this is precisely the context of  
\cite{Castagnoli} and \cite{Frittelli}, the earliest contributions on the collapse to the mean.
These challenge \cite{Wang2000,Wang2002}'s harmonisation of financial and insurance pricing, which builds on \cite{Wangetal}'s axiomatisation of insurance pricing rules as law-invariant Choquet integrals. 
They effectively question whether a pricing functional can sensibly be law invariant
by exposing the fundamental tension between the law invariance of a sublinear or convex pricing rule and the presence of nontrivial frictionlessly priced payoffs.
The only admissible pricing rule able to cope with both requirements is the expectation, the actuarially fair price, yet it does not contain any safety loading.

In other words, counter to the modeller's intuition, a functional derived from seemingly reasonable axioms may discard all distributional information beyond the mean of its argument.
The repercussions of this phenomenon are not limited to insurance pricing; for example, \cite[Section 5]{Bellini} interpret this result as the underpinning of the partial law invariance encountered in market-consistent valuation and apply it to capital requirements based on eligible assets in the spirit of \cite{Artzner}.
In \cite{Liebrich1} and \cite{Liebrich2}, collapse-to-the-mean results (implicitly) serve as important tools in solving optimal risk sharing problems involving eligible assets or consistent risk measures. 
A further example is \cite{Centrone} in the context of capital allocation.

A growing body of literature has subsequently shown how robust the collapse-to-the-mean phenomenon is, and has streamlined and generalised the associated mathematical results.
While \cite{Bipolar} casts them in terms of the bipolar theorem, \cite{Bellini,Chen,Liebrich} and \cite{Amarante} have extended the analysis to more general spaces and classes of functionals; relaxed countable additivity of the reference measure; or replaced local linearity with conditions requiring, e.g., that risk is not increased when suitable prospects are added to the initial position.
These extensions broaden the economic implications of the phenomenon and reveal further limits to meaningful risk management wherever law invariance is imposed.

Beyond these contributions dedicated specifically to the collapse to the mean, the debate connects to further strands of the risk management literature.
A collapse to the mean can be interpreted as a failure of the ``sensitivity to large losses" property of  \cite{herdegen2024}. The conditions forcing expectation-basedness can be viewed through the lens of risk reducers, as studied for the convex order by \cite{Cheung,He} and for the Expected Shortfall by \cite{herdegen2025}. 
Moreover, expectation-basedness coincides with the ``risk neutrality" of \cite{Maccheroni}, whose global properties of dependence neutrality and neutrality to full insurance are equivalent to the local conditions appearing in collapse-to-the-mean results.

A study closely related to the present paper  is \cite{Liebrich}. 
That work treats the collapse to the mean in as general a framework as possible and avoids assumptions of convexity or monotonicity in the almost-sure order where possible. 
However, it does not provide a unified mathematical exposition with a single main result from which the individual case studies follow as corollaries. 
As a consequence, its findings may appear fragmented or technical.

In the present paper, we address this limitation and contribute to the literature as follows. 
First, we unify the bulk of existing results without imposing convexity or monotonicity assumptions.
Second, we derive the collapse-to-the-mean results not from law invariance per se, but from consistency with the convex order. Its concise  characterisation in \cite{Shaked} serves as a simple and transparent mathematical foundation, allowing to avoid dual methods and streamline proofs. 
Third, this simplified perspective naturally leads to a generalisation of the existing results that covers the quasi-star-shaped functionals recently introduced by \cite{Hanetal} and even potentially incomplete law-invariant preferences.
Fourth, we draw explicit connections between the collapse to the mean 
and risk reducers, sensitivity to large losses, and insurance propensity, thereby involving a body of literature outside of the scope of preceding contributions and improving the relevance for risk management in actuarial and financial contexts.

The paper is organised as follows. Section~\ref{sec:preliminaries} recalls the preliminaries, reviews the relevant existing results unified by the present approach, and discusses the connection between the collapse to the mean and recent contributions in the risk management literature. Section~\ref{sec:main} presents and thoroughly discusses our main unifying results. The mathematical analysis and proofs of the main results are provided in Section~\ref{sec:proofs}.

\section{Preliminaries and background}\label{sec:preliminaries}

\subsection{Preliminaries}\label{sec:setting}

Throughout the paper, $(\Omega, \mathcal{F}, \mathbb{P})$ denotes an underlying atomless probability space and $L^\infty$ ($L^1$, respectively) the space of equivalence classes of bounded (finite-mean) random variables.
These spaces carry the usual almost-sure (a.s.) order to compare suitable random variables. 
Moreover, a random variable $X$ is {\em nonconstant} if $\P(X=c)<1$ for all constants $c\in\R$. 

The convex order on $L^1$---denoted by $\cx$---asserts that $Y$ dominates $X$ ($X \cx Y$) if
\begin{equation}\label{def conv}
\E[u(X)] \leq \E[u(Y)] \quad \text{for all convex functions }u \colon \mathbb{R} \to \mathbb{R}.
\end{equation}
For mathematical and conceptual challenges with the convex order beyond $L^1$, we refer to \cite{Cote}.

In this paper, we deal with functionals on a given set $\CX\subseteq L^1$ that take values in the extended real line, and with preference relations (preorders) $\peq$ on $\CX$, i.e., transitive and reflexive but not necessarily complete binary relations.
A preference relation $\peq$ with symmetric part $\sim$ is said to be:
\begin{enumerate}[label=(\alph*)]
\item  {\em law invariant} if two random variables $X,Y\in\CX$ with the same distribution under $\P$ satisfy $X\sim Y$. 
\item \emph{$\cx$-consistent} (or {\em Schur convex}) if $X\succeq Y$ whenever $X,Y\in\CX$ satisfy $X \cx Y$.
\item  
{\em expectation-based} if $X\sim Y$ whenever  $X,Y\in\CX$ have the same mean under $\P$, i.e., they satisfy $\E[X]=\E[Y]$.
\end{enumerate}
Clearly, the implications (c)$\implies$(b)$\implies$(a) hold.

Every functional $\ph\colon\CX\to[-\infty,\infty]$ induces a preference relation $\peq_\ph$ on $\CX$ by setting 
\begin{equation}\label{eq:peqph}
    X\peq_\ph Y\quad\iff\quad \ph(X)\ge\ph(Y).
\end{equation}
If $\ph$ measures the risk of prospects in $\CX$, $\peq_\ph$ encodes the preference for less risk. 
Via \eqref{eq:peqph}, it is often easy to translate properties of functional $\ph$ to the preference relation $\peq_\ph$ and vice versa, and we shall do so regularly. For example, a functional $\ph$ is $\cx$-consistent if $X\cx Y$ implies $\ph(X)\le \ph(Y)$.

Regarding the domain of definition of the functionals studied in this paper, we impose \cite[Assumption 2.1]{Liebrich} and fix a pair $(\CX,\CX^*)$ of subspaces of $L^1$ with the following properties. First, both $\CX$ and $\CX^*$ contain $L^\infty$ as subspace. Second, both spaces are law-invariant, i.e., whether $X\in L^1$ belongs to one of the spaces or not depends only  on its distribution under $\P$.
Third, they satisfy
\[\CX\cdot\CX^*:=\{XY\mid X\in\CX,\,Y\in\CX^*\}\subseteq L^1.\]
This permits the definition of the locally convex $\sigma(\CX,\CX^*)$-topology, the weakest linear topology on $\CX$ making each functional $X\mapsto\E[XY]$, $Y\in\CX^*$, continuous.

\begin{example}
    Examples of the pair  $(\CX,\CX^*)$ include: 
    \begin{enumerate}[label=(\arabic*)]
        \item For some $p,r\in [1,\infty]$ with $\frac 1 p+\frac 1 r\le 1$,  $\CX=L^p$ and $\CX^*=L^r$. This case includes the classical weak topologies on $L^p$-spaces, $p\in[1,\infty)$. 
        \item $(\CX,\CX^*)$ is a pair of Orlicz spaces. Suppose $\Phi\colon [0,\infty)\to[0,\infty)$ is convex, increasing, and satisfies $\Phi(0)=0$ as well as $\lim_{t\to\infty}\frac{\Phi(t)}t=\infty$. If $\overline \Psi$ is another function with these properties that additionally satisfies 
        \[\overline \Psi(t)\ge \sup_{s\ge 0}\{st-\Phi(s)\},\quad t\ge 0,\]
        then the pair of Orlicz spaces $(L^\Phi,L^{\overline\Psi})$ fits our description. We refer to \cite{Orlicz} for details. 
        \item $\CX$ is a rearrangement-invariant space and  $\CX^*$ its associate space or Köthe dual, i.e., those norm-continuous linear functionals of shape $X\mapsto\E[YX]$ for an  integrable density $Y$; see \cite[Section 1.1]{Automatic} for details.
    \end{enumerate}
\end{example}

In this situation, a functional $\ph\colon\CX\to[-\infty,\infty]$ is said to be:
\begin{enumerate}[label=(\alph*)]
\addtocounter{enumi}{3}
    \item {\em proper} if $\ph>-\infty$ and the {\em (effective) domain} of $\ph$ 
    \[\dom(\ph):=\{X\in\CX\mid \ph(X)\in\R\}\]
    is nonempty.
    \item {\em quasiconvex} if all sublevel sets 
    \[\mathcal L_\alpha(\ph):=\{X\in\CX\mid \ph(X)\le \alpha\},\quad \alpha\in[-\infty,\infty],\] 
    are convex. 
    \item {\em $\tau$-lower semicontinuous} ($\tau$-lsc) for a topology $\tau$ on $\CX$ if all sublevel sets $\mathcal L_\alpha(\ph)$, $\alpha\in[-\infty,\infty]$, are $\tau$-closed.
\end{enumerate}
Using \cite[Theorem 3.6]{General}, every proper, quasiconvex, $\sigma(\CX,\CX^*)$-lsc and law-invariant functional is for example automatically $\cx$-consistent. A preference relation $\peq$ is 
\begin{enumerate}[label=(\alph*)]\addtocounter{enumi}{6}
    \item {\em $\tau$-upper semicontinuous} ($\tau$-usc) if the upper level set $\{Y\in\CX\mid X\peq Y\}$ is $\tau$-closed for all $X\in\CX$. 
\end{enumerate}

On the dual space $\CX^*$, a proper functional $\ph$ defined on $\CX$ induces its {\em convex conjugate}, 
\begin{equation}\label{def conj}
    \ph^*(Y):=\sup_{X\in\CX}\{\E[YX]-\ph(X)\},\quad Y\in\CX^*.
\end{equation}
The set $\dom(\ph^*):=\{Y\in\CX^*\mid \ph^*(Y)<\infty\}$ is the domain of the convex conjugate.
We shall go beyond this concept though when talking about the ``dual collapse'' below and will be interested in dual elements $Y\in\CX^*$ with the property that, for a given $\alpha\in\R$, there exists $\beta\in\R$ such that 
\[\ph(X)\le \alpha\quad\implies\quad \E[YX]\le \beta.\]
To this end, we introduce for functionals $\ph\colon\CX\to[-\infty,\infty]$ the associated function $R_\ph\colon[-\infty,\infty]\times\CX^*\to[-\infty,\infty]$ by
    \begin{equation}\label{eq:Rph}R_\ph(\alpha,Y):=\sup\{\E[YX]\mid \ph(X)\le \alpha\}.\end{equation}
The relationship between $\ph^*$ and $R_\ph$ is illustrated in Example~\ref{ex} below. 
For a preference relation $\peq$, we modify this definition slightly as 
\begin{equation}\label{Rpeq}
R_\peq(X,Y):=\inf\{\E[YZ]\mid Z\succeq X\},\quad (X,Y)\in\CX\times\CX^*.\end{equation}
Consequently, $R_\peq(X,Y)>-\infty$ if and only if the evaluation of prospects $Z$ at least as desirable as $X$ with $\E[Y\cdot]$ does not become arbitrarily negative.
The relation between \eqref{eq:Rph} and \eqref{Rpeq} is elucidated by the following lemma:

\begin{lemma}\label{lem:relation}
    Given a functional $\ph\colon\CX\to[-\infty,\infty]$, $c\in\R$ and $Y\in\CX^*$, we have 
\[R_{\peq_\ph}(c,Y)=-R_\ph(\ph(c),-Y).\]
\end{lemma}
    
Finally, we denote by $q_X$ a quantile function of $X\in L^1$ and by 
\[\es_p(X):=\begin{cases}\frac1{1-p}\int_p^1q_X(s)ds&\quad\text{if }0\le p<1,\\
\lim_{q\uparrow 1}\es_q(X)&\quad\text{if }p=1,\end{cases}\]
the {\em Expected Shortfall} of $X$ at level $p\in[0,1]$.
For fixed $X\in L^1$, the function $[0,1]\ni p\mapsto\es_p(X)$
is nondecreasing.

\subsection{Background}\label{sec:background}

Here, we sketch the background of this paper and recall collapse to the mean results that can be found in the recent literature, paving the way to the unified treatment in Section~\ref{sec:main}.
We begin with \cite[Theorem 5.3]{Liebrich}.

\begin{theorem}\label{thm:LM1}
    For a proper, quasiconvex, $\sigma(\CX,\CX^*)$-lsc and law-invariant functional $\ph\colon \CX\to(-\infty,\infty]$, the following are equivalent: 
    \begin{enumerate}[label=\tn{(\arabic*)}]
        \item $\ph$ is expectation-based.
        \item There exists a nonconstant $Z\in\CX$ with mean zero such that, for all $X\in\CX$,  
        \[\ph(X+tZ)\le\ph(X),\quad t\ge 0.\]
        \item For all $X\in\CX$ there exists a nonconstant $Z\in\CX$ with mean zero such that
\[\ph(X+tZ)\le\ph(X),\quad t\ge 0.\]
        \item For all $\alpha\in\ph(\R)$, 
\[\dom\big(R_\ph(\alpha,\cdot)\big)\subseteq\R.\] 

    \end{enumerate}
\end{theorem}

Conditions (2) and (3) describe local nonexpansiveness and are best understood when the functional $\ph$ measures the financial risk of random prospects:
Given any background risk $X$, one can add $Z$ at an arbitrary exposure without increasing the measured risk. 
Notably though, most collapse-to-the-mean results have been proved for convex rather than merely quasiconvex functionals; see the discussion in \cite{Bellini}. 
As an alternative to the nonexpansiveness conditions in Theorem~\ref{thm:LM1}, \cite[Theorems 4.5 \& 4.7]{Bellini} and \cite[Theorems 5.1 \& 5.2]{Liebrich} show that
variants of the following local linearity conditions imply that a convex functional $\ph$ is expectation-based or, more specifically, an affine function of the mean: There exists a nonconstant $Z\in\dom(\ph)$ and $a\in\R$ such that either
\begin{equation}\label{eq:affine}
    \ph(tZ)=\ph(0)+ta,\quad t\in\R,
\end{equation}
or 
\begin{equation}\label{eq:TI}
    \ph(X+tZ)=\ph(X)+ta,\quad (X,t)\in\CX\times\R.
\end{equation} 

\begin{lemma}
    For a proper, convex and $\sigma(\CX,\CX^*)$-lsc functional $\ph$,  \eqref{eq:affine} and \eqref{eq:TI} are equivalent.
\end{lemma}
\begin{proof}
Clearly, \eqref{eq:TI} implies \eqref{eq:affine}. 
Conversely, \eqref{eq:affine} implies that $\ph(0)=\ph(Z-Z)=\ph(Z)-a\in\R$. Setting $\lambda_n=\frac 1 n$ and using that $\ph$ is lsc and convex, we have for all $X\in\CX$ and $t\in\R$ that
    \begin{align*}\ph(X+tZ)&\le \liminf_{n\to\infty}\ph\big((1-\lambda_n)X+\lambda_n\tfrac{tZ}{\lambda_n}\big)\le \lim_{n\to\infty}(1-\lambda_n)\ph(X)+\lambda_n\ph(0)+at=\ph(X)+at.
\end{align*}
Writing $X$ in the final term of the preceding estimate as $X+tZ-tZ$ and using the estimate itself, we obtain 
\begin{center}
    $\ph(X)+at\le \ph(X+tZ)-at+at=\ph(X+tZ).$
\end{center}
\end{proof}

Moreover, \eqref{eq:TI} has been studied in the context of quasiconvex functionals by \cite[Theorem 5.5]{Liebrich}.
While the relationship of \eqref{eq:affine} and \eqref{eq:TI} to Theorem~\ref{thm:LM1} is not clear {\em a priori}, we shall elucidate it in Theorem~\ref{thm1} below, thereby unifying the convex and the quasiconvex cases. 

Another key result of \cite{Liebrich}---that has been applied in the context of risk sharing by \cite{Liebrich2}---concerns {\em consistent risk measures}. The latter term coined by \cite{Consistent} refers to---not necessarily (quasi)convex---monetary risk measures which are consistent with the {\em increasing convex order} relation between random variables. 
We say that $X,Y\in L^1$ are in increasing convex order ($X\ssd Y$) if 
\[\E[u(X)] \leq \E[u(Y)] \quad \text{for all convex and nondecreasing functions }u \colon \R\to \R.\]
The $\ssd$-consistency of a functional is defined analogously to $\cx$-consistency. 
The next result is \cite[Theorem 5.7]{Liebrich}. 

\begin{theorem}\label{thm:LM2}
    Suppose a  $\sigma(\CX,\CX^*)$-lsc functional $\ph\colon \CX\to(-\infty,\infty]$ has the following properties:
    \begin{enumerate}[label=\tn{(\alph*)}]
        \item $\ph(0)=0$.
        \item $\ph$ is cash-additive, i.e., 
    \begin{equation}\label{eq:cash}\ph(X+c)=\ph(X)+c,\quad X\in\CX,\,c\in\R.\end{equation}
        \item $\ph$ is $\ssd$-consistent.
    \end{enumerate}
    Then the following are equivalent: 
    \begin{enumerate}[label=\tn{(\arabic*)}]
        \item $\ph=\E[\cdot]$.
        \item There exists a nonconstant $Z\in\CX$ such that $\ph(tZ)=t\ph(Z)$, $t\in\R$. 
        \item There exists a nonconstant $Z\in\CX$ with mean zero such that 
        $\sup_{t\ge 0}\ph(tZ)= 0$. 
    \end{enumerate}
    Any of these statements implies
    \[\dom(\ph^*)=\{1\},\]
    and this implication is an equivalence if $\ph$ is additionally star-shaped, i.e., for all $X\in \CX$ and all $\lambda\in(0,1)$, 
    \[\ph(\lambda X)\le \lambda\ph(X).\]
\end{theorem}

Theorem~\ref{thm:LM2} clarifies our distinction between \emph{primal} and \emph{dual} collapse results mentioned above.
Primal collapse results derive that a law-invariant functional is expectation-based from its local behaviour on the ``primal space'' $\CX$. Examples include the linearity of $t \mapsto \ph(tZ)$ in (2) and the nonexpansiveness condition in (3). 
For dual collapse results, note that it is constants $c \in \CX^*$ that give rise to the functional $c\E[\cdot]$ under the $(\CX,\CX^*)$ pairing. 
Hence, 
dual collapse results equate expectation-basedness with the condition that a dual description of the functional in question---such as the conjugate $\ph^*$---defined on the ``dual space" $\CX^*$ only needs to be evaluated in constant random variables. 
For instance, Theorem~\ref{thm:LM2} shows that expectation-basedness is tantamount to $\dom(\ph^*)=\{1\}$.

The mathematical techniques used in \cite{Bellini} for convex functionals are quite different from the ones employed in \cite{Liebrich}. 
Also, the proofs of Theorems~\ref{thm:LM1} and~\ref{thm:LM2} in the latter work differ substantially.
Theorem~\ref{thm:LM1} relies on dual representations of level sets of quasiconvex functionals.
Consistent risk measures, however, are generally not (quasi)convex and therefore lack such a dual representation.
Accordingly, the proof of Theorem~\ref{thm:LM2} follows the approach of \cite{Consistent} and uses their representation of consistent risk measures on $L^\infty$ as lower envelopes of adjusted Expected Shortfalls.

The main part of the paper, Section~\ref{sec:main}, will show that generalised versions of the results discussed above can be derived from a single streamlined and transparent mathematical foundation. 
On closer inspection, our reasoning turns out to be more akin to the proof of Theorem~\ref{thm:LM2} than that of Theorem~\ref{thm:LM1}.

\subsection{Related literature}\label{sec:literature}

While there are numerous contributions devoted specifically to the col\-lapse-to-the-mean phenomenon, interesting connections and comparisons can also be drawn with other strands of literature that tackle different problems, but are related in spirit. 
This will be the focus of the present subsection. {\em Inter alia}, we aim to highlight the relevance of collapse-to-the-mean studies for risk management more broadly.  

\subsubsection*{Risk reducers}
\cite{Cheung} introduced the notion of a \emph{risk reducer}, later adopted by \cite{He}.
For an integrable loss $X$, a random variable $Z\in L^1$ is a risk reducer if
\begin{equation}\label{eq:risk reducer1}
    X+Z \cx X+\E[Z],
\end{equation}
that is, if adding random net loss $Z$ to initial loss $X$ is preferred in convex order to the addition of the deterministic loss $\E[Z]$.
Equivalently, \eqref{eq:risk reducer1} can be expressed as 
\begin{equation}\label{eq:risk reducer2}
    X+Z-\E[Z] \cx X.
\end{equation}
Both \eqref{eq:risk reducer1} and \eqref{eq:risk reducer2} express unanimous preference among risk-averse expected-utility (EU) agents: faced with background risk $X$, they either prefer adding risk $Z$ to adding its mean $\E[Z]$, or they prefer adding the zero-mean risk $Z-\E[Z]$ to leaving the background risk unchanged. 
While \cite{Cheung} focus on the existence of risk reducers countermonotonic with the background risk, \cite{He} argue that the countermonotonicity assumption is incompatible with common insurance applications and should therefore be dropped.

This intuition is reminiscent of the  collapse-to-the-mean literature and our main results.
Theorem~\ref{thm:LM1} and its generalisation Theorem~\ref{thm1} below, however, study the existence of risks $Z$ with zero mean for which {\em arbitrarily large} exposures do not increase the risk of a given---often deterministic---baseline loss.
Moreover, risk is not evaluated via a stochastic order representing broad consensus among EU agents, but via a specific functional $\ph$ or an idiosyncratic preference relation $\peq$ whose structural properties drive the results.
Finally, note the contrast between \eqref{eq:risk reducer2} and Lemma~\ref{lem}, which shows that for all simple $X,Z$ with $\E[Z]=0$, and for sufficiently large exposures $t$,
\[
    X \cx \E[X] + tZ.
\]
That is, although the mean is unanimously preferred to the risky loss $X$, this preference reverses once the fluctuations of $X$ around its mean are replaced by $tZ$ with sufficiently large variance.

Independent of these contributions, risk reducers have also been introduced recently by \cite{herdegen2025}.
A point of contact between their analysis and the present study is that risk reduction is understood relative to a specific law-invariant functional, in their case  Expected Shortfall applied to payoffs instead of losses.
In contrast to \cite{Cheung} and \cite{He}, a risk reducer is a payoff $Z$ such that, {\em for all $X$},
\[\es_p(X+Z)\le\es_p(X).\]
This immediately implies that, for all $X$ and all $t>0$, 
\[\es_p(X+tZ)\le \es_p(X).\]
which parallels Theorem~\ref{thm:LM1}.
In their terminology, Theorem~\ref{thm:LM1} rules out the existence of nontrivial zero-mean risk reducers.
Their focus, however, is different: they investigate risk reducers that arise as the payoff of a fully leveraged (i.e., zero-cost) portfolio in an underlying financial market, which therefore constitute a type of arbitrage opportunity.

\subsubsection*{Sensitivity to large losses}
Another related contribution is \cite{herdegen2024}.
In the spirit of the calibration theorem of \cite{Rabin}, the authors examine whether risk (utility) functionals correctly assess random future payoffs with downside risk.
This requires that ``positions with sufficiently large loss peaks [should be] deemed unacceptable'', i.e., should be assigned positive risk (or negative utility) once the exposure is sufficiently large.

Formally, if a functional $\ph$ measures the risk of a payoff, it is {\em sensitive to large losses} on a subset $\mathcal M$ of its domain if, for all $X\in\mathcal M$,  
\[\P(X<0)>0\quad\implies\quad\ph(\lambda X)>0\text{ for all }\lambda>0\text{ large enough}.\]
For star-shaped risk functionals, this is equivalent to 
\begin{equation}\label{eq:sensitive}\P(X<0)>0\quad\implies\quad\sup_{\lambda>0}\ph(\lambda X)=\infty;\end{equation}
see \cite[Theorem 3.9]{herdegen2024}. 
Typical choices of $\mathcal M$ are ``pure losses'' ($X\le 0$ a.s.) and ``expected losses'' ($\E[X]\le 0$).
If we take $\mathcal M$ to be zero-mean risks,
\[\mathcal M=\{X\mid \E[X]=0\},\]
then \eqref{eq:sensitive} negates condition (3) in Theorem~\ref{thm:LM1} (there we need to set $X=0$).
In this interpretation, expectation-basedness follows from a lack of sensitivity to large losses.

\subsubsection*{Insurance propensity} 

\cite{Maccheroni} study risk aversion in the context of potentially incomplete and/or intricate individual preferences. Their main assertion is that risk aversion can be fully characterised by the propensity to choose specific types of insurance contracts.
Mathematically, choices are governed by transitive and law-invariant preference relations $\peq$ defined on a space of admissible random wealth variables; admissibility is tantamount to possessing finite absolute moments of all orders at least 1.
This framework provides great flexibility in modelling individual preferences and encompasses many of the preference models considered in the literature; see \cite[pp.~1601--1602]{Maccheroni}.
In Theorem~\ref{thm2}, we adopt this perspective and work with preference relations instead of functionals.

\cite[Proposition 2]{Maccheroni} characterises risk neutrality, i.e., the fact that 
\[X\sim\E[X]\]
holds for all wealths $X$. 
This can easily be seen to be equivalent to $\peq$ being expectation-based as defined in Section~\ref{sec:preliminaries}. Their result shows {\em inter alia} that risk neutrality is equivalent to:
\begin{enumerate}[label=(\arabic*)]
    \item {\em Dependence neutrality}:
    For all admissible random variables $X,Y,Z$, 
    \begin{center}
        $X+Y \sim X+Z$
    \end{center}
    provided $Y$ and $Z$ have the same distribution under $\P$.

    \item {\em Neutrality to full insurance}:
    Suppose $X,Y,Z$ are admissible random variables, $Y=-X-\pi$ for some $\pi\in\R$---that is, $Y$ provides full insurance for random wealth $X$---and $Z$ has the same distribution as $Y$ under $\P$. Then $X+Y \sim X+Z$.
\end{enumerate}

From a risk management perspective, it seems  counterintuitive to be indifferent between full insurance eliminating all residual randomness on one hand, and an arbitrary alternative insurance contract with equal distribution but imperfect hedging power for $X$ on the other hand. 
Therefore, the equivalence between risk neutrality and neutrality to full insurance reiterates the observation that expectation-basedness is a deficiency in actuarial or risk management contexts. 

While dependence neutrality is reminiscent of the conditions appearing in the collapse to the mean results presented here, in particular Theorem~\ref{thm2} below, it should be noted that its logical structure differs. 
It is formulated as a universal ``for all'' condition, whereas the characterisations of expectation-basedness in the results here contain existence statements in their antecedents.


The different mathematical nature of their results is also reflected by the mathematical analysis in~\cite{Maccheroni} being substantially different from our approach. 
The authors do not rely on the convex order, but instead prove and exploit a deep structural result which is easy to state: 
an admissible random variable $X$ has zero mean if and only if there exist two equidistributed admissible random variables $Y$ and $Z$ such that $X$ coincides in distribution with $Y-Z$; see \cite[Lemma~1 \& Theorem~6]{Maccheroni}.
Nevertheless, its proof is highly involved. 
While this result is an independent contribution of~\cite{Maccheroni} which will undoubtedly have more applications in the future, our aim here is to keep the analysis as elementary and self-contained as possible.

\section{Main results}\label{sec:main}

\subsection{The collapse of functionals and preferences}

We begin our discussion of the collapse to the mean with functionals $\ph$ defined on a space $\CX$ as described in Section~\ref{sec:setting}. 
Theorem~\ref{thm1} below sets out the relationships among the following statements about such functionals:

\begin{enumerate}[label=(F\arabic*),ref=\tn{(F\arabic*)}]
    \item\label{F1} $\ph$ is expectation-based.

\item\label{F2} There exists a nonconstant zero-mean $Z\in\CX$ such that, for all $X\in\CX$,  
        \begin{equation}\label{eq:sup}\ph(X+tZ)\le\ph(X),\quad t\ge 0.\end{equation}
        
\item\label{F3} There exists a nonconstant zero-mean $Z\in\CX$ such that, for all $c\in\R$, 
\begin{equation}\label{eq:sup const}\ph(c+tZ)\le\ph(c),\quad t\ge 0.\end{equation}

\item\label{F4} For all $X\in\CX$ there exists a nonconstant zero-mean $Z\in \CX$, possibly depending on $X$, such that \eqref{eq:sup} holds. 

\item\label{F5} For all $c\in\R$ there exists a nonconstant zero-mean $Z\in \CX$, possibly depending on $c$, such that \eqref{eq:sup const} holds.
\item\label{F6} 
There exists a nonconstant $Z\in\CX$ and $a\in\R$ such that, for all $X\in\CX$ and $t\in\R$,
\begin{equation*}
    \ph(X+tZ)=\ph(X)+ta.
\end{equation*}  
\item\label{F7}
For all $\alpha\in\ph(\R)$, 
\begin{equation}\label{eq:R}\dom\big(R_\ph(\alpha,\cdot)\big)\subseteq\R.\end{equation} 
\end{enumerate}

Statements~\ref{F1}--\ref{F6} constitute the ``primal component'' of Theorem~\ref{thm1}, equivalent local properties of $\ph$'s behaviour on $\CX$ that render said functional expectation-based.
\ref{F2}--\ref{F5} describe the nonexpansiveness of $\ph$ along a suitable zero-mean direction $Z$. \ref{F6} is a local linearity condition. 

The theorem's ``dual component'', statement~\ref{F7}, is a condition formulated on the dual space $\CX^*$  that appears {\em verbatim} in Theorem~\ref{thm:LM1} and asserts that $R_\ph(\alpha,Y)$ can only be finite for constant dual elements $Y$.
If $\ph$ has the additional property of \emph{quasi-star-shapedness}, statement \ref{F7} implies expectation-basedness.
Quasi-star-shapedness has recently been introduced by \cite{Hanetal}
and describes that a functional $\ph\colon\CX\to(-\infty,\infty]$ satisfies   
\begin{equation}\label{def:qss}\text{for all }X\in\CX,\,t\in\R,\text{ and }\lambda\in(0,1):\quad \ph(\lambda X+(1-\lambda)t)\le \max\{\ph(X),\ph(t)\}.\end{equation}
This clearly generalises quasiconvexity, which would require the inequality $\ph(\lambda X+(1-\lambda)Y)\le \max\{\ph(X),\ph(Y)\}$ to hold for all $X,Y\in\CX$, i.e., also for nonconstant arguments.

Theorem~\ref{thm1} also covers convex functionals explicitly for the sake of completeness. 

\begin{theorem}\label{thm1}
Let $\ph\colon\CX\to(-\infty,\infty]$ be a proper, $\sigma(\CX,\CX^*)$-lsc and $\cx$-consistent functional.
\begin{enumerate}[label=\tn{(\arabic*)}]
\item Statements \ref{F1}--\ref{F6} are all equivalent and imply statement \ref{F7}. 
\item If the nonconstant $Z\in\CX$ appearing in \ref{F6} additionally satisfies $\E[Z]\neq 0$, then there exist $\alpha,\beta\in\R$ such that 
\[\ph=\alpha\E[\cdot]+\beta.\]
\item Statements \ref{F1}--\ref{F7} are all equivalent if $\ph$ is additionally quasi-star-shaped as defined by \eqref{def:qss}.
\item 
If $\ph$ is convex, \ref{F1}--\ref{F7} are equivalent to
\begin{equation}\label{eq:phi*}\dom(\ph^*)\subseteq\R.\end{equation}
\end{enumerate}
\end{theorem}

Before commenting on Theorem~\ref{thm1} in Section~\ref{sec:comments}, we state a corollary on cash-additive functionals obtained by a direct application of this theorem.

\begin{corollary}\label{cor}
    Consider the following statements about a proper, $\sigma(\CX,\CX^*)$-lsc and $\cx$-consistent functional $\ph\colon \CX\to(-\infty,\infty]$ with the cash-additivity property \eqref{eq:cash}.
    \begin{enumerate}[label=\tn{(C\arabic*)},ref=\tn{(C\arabic*)}]
        \item\label{C1} $\ph=\ph(0)+\E[\cdot]$.
        \item\label{C2} There exists a nonconstant $Z\in\CX$ with mean zero such that 
        $\sup_{t\ge 0}\ph(tZ)\le\ph(0)$.
        \item\label{C3} There exists a nonconstant $Z\in\CX$ and $a\in\R$ such that $\ph(tZ)=\ph(0)+at$, $t\in\R$. 
        \item\label{C4} $\dom(\ph^*)\subseteq\R$.
    \end{enumerate}
    Then:
    \begin{enumerate}[label=\tn{(\arabic*)}]
    \item Statements \ref{C1}--\ref{C3} are all equivalent and imply \ref{C4}. 
    \item 
    If $\ph(0)=0$ and $\ph$ is star-shaped, \ref{C4} also implies \ref{C1}.
    \end{enumerate}
\end{corollary}

Inspired by \cite{Maccheroni}, we also recast Theorem~\ref{thm1} within the more general framework of preferences over space $\CX$, a rich setting that encompasses many preference models in the literature as discussed above.
By the inference of preference relations from functionals in \eqref{eq:peqph}, Theorem~\ref{thm2} subsumes large parts of Theorem~\ref{thm1} while also accommodating cases in which we remain agnostic about the existence of a numerical representation (which would, for instance, follow from Cantor’s Theorem~\cite[Theorem~6.1]{Gilboa}) or even about completeness of the preferences. 
Theorem~\ref{thm2} relates the following statements about a preference relation $\peq$: 
 
\begin{enumerate}[label=\tn{(P\arabic*)},ref=\tn{(P\arabic*)}]
\item\label{P1} $\peq$ is expectation-based.

\item\label{P2} There exists a nonconstant $Z\in\CX$ with mean zero such that, for all $X\in\CX$, 
\begin{equation}\label{eq:supPREF}X\peq X+tZ,\quad t\ge 0.\end{equation}
        
\item\label{P3} There exists a nonconstant $Z\in\CX$ with mean zero such that, for all $c\in\R$, 
\begin{equation}\label{eq:supPREF const}c+tZ\sim c,\quad t\ge 0.\end{equation}

\item\label{P4} For all $X\in\CX$ there exists a nonconstant $Z\in \CX$ with mean zero such that \eqref{eq:supPREF} holds. 

\item\label{P5} For all $c\in\R$, there exists a nonconstant $Z\in \CX$ with mean zero such that \eqref{eq:supPREF const} holds. 

\item\label{P6} For all $c\in\R$, 
\begin{center}$R_\peq(c,Y)>-\infty\quad\implies\quad Y\in\R$.\end{center}
\end{enumerate}

Again, the $Z$ appearing in \ref{P4} and \ref{P5} may depend on $X$ or $c$, respectively.

\begin{theorem}\label{thm2}
Suppose $\peq$ is a $\cx$-consistent and $\sigma(\CX,\CX^*)$-usc preference relation on $\CX$.
\begin{enumerate}[label=\tn{(\arabic*)}]
    \item Statements \ref{P1}--\ref{P5} are all equivalent and imply statement \ref{P6}.
    \item Implication $\ref{P6}\implies\ref{P1}$ holds if, additionally, for all $X\in\CX$, $x,y\in\R,$ and $\lambda\in(0,1)$, 
\begin{equation}\label{eq:starPREF}x\peq X\text{ and }x\peq y\quad\implies\quad x\peq \lambda X+(1-\lambda)y.\end{equation}
\end{enumerate}
\end{theorem}

\subsection{Discussion of the main results}\label{sec:comments}

First, it is important to note that the primal part of Theorem~\ref{thm1} does not require any (quasi)convexity assumption.
Hence, Theorem~\ref{thm:LM1} and the results about convex functionals discussed in its context are merely special cases of Theorem~\ref{thm1}.
Moreover, consistent risk measures as introduced in Theorem~\ref{thm:LM2} are $\cx$-consistent, which means that Corollary~\ref{cor} covers Theorem~\ref{thm:LM2} as a special case.

The dual collapse to the mean in Theorem~\ref{thm1}---the equivalence between statements~\ref{F1} and \ref{F7} under quasi-star-shapedness---is genuinely new.
Notably, most existing dual collapse results have been obtained via the Fenchel-Moreau representation of convex functionals.
However, Fenchel-Moreau type representation results for \emph{quasi-star-shaped} functionals are to our knowledge not available.
Instead, the proof of the implication~$\ref{F7}\implies\ref{F1}$ below is inspired by the strategy employed in \cite[Theorem~5.7, (iv)$\implies$(iii)]{Liebrich}.
It avoids imposing stronger structural properties on the functional, such as (quasi)convexity, cash-additivity, or monotonicity with respect to the a.s.\ order, that is, the requirement that $\ph(X)\le \ph(Y)$ whenever $X\le Y$ a.s.
Nevertheless, $\sigma(\CX,\CX^*)$-lsc quasiconvex functionals and consistent risk measures are covered as special cases, and Theorem~\ref{thm1} also includes the dual statements of Theorems~\ref{thm:LM1} and~\ref{thm:LM2}.

Another---perhaps obvious---consequence of Theorem~\ref{thm1} is that any proper, $\sigma(\mathcal X,\mathcal X^*)$-lsc, $\cx$-consistent, and quasi-star-shaped functional satisfying condition~\ref{F7} must in fact be quasiconvex.
This follows immediately from the fact that such functionals are expectation-based.

The results are also sharp, as illustrated already by \cite[Example E.3]{Liebrich} and \cite[Example 2.4]{Liebrich2}.

We now zoom in on some of the technical assumptions involved in the preceding results and illustrate first the relationship between the functional $\ph$, the function $R_\ph$, and the convex conjugate $\ph^*$ in three case studies.

\begin{example}\label{ex}\,
\begin{enumerate}[label=(\arabic*)]
    \item For a proper and cash-additive $\ph$ consider $\mathcal A_\ph:=\mathcal L_0(\ph)\subsetneq \CX$, a set which is usually called {\em acceptance set} in the context of (monetary) risk measures. 
    It is of crucial importance for cash-additive functionals as they can be recovered from their acceptance set by the formula 
    \[\ph(X)=\inf\{m\in\R\mid X-m\in\mathcal A_\ph\}.\]
    The set $\mathcal A_\ph$ induces on $\CX^*$ the functional 
    \[\sigma(Y):=\sup_{X\in\mathcal A_\ph}\E[YX]\]
    whose effective domain is the so-called barrier cone of $\mathcal A_\ph$. In view of cash-additivity of $\ph$, we have for all $\alpha\in\R$ and $Y\in\CX^*$ that
    \[R_\ph(\alpha,\cdot)=\sigma(\cdot)+\alpha\E[\cdot]\quad\text{and}\quad\ph^*(Y)=\begin{cases}\sigma(Y)&~\text{if }\E[Y]=1,\\[-0.6ex]
    \infty&~\text{otherwise}.\end{cases}\]
    Consequently, each $\dom(R_\ph(\alpha,\cdot))$ agrees with said barrier cone, which can be rephrased for arbitrary $Y\in\CX^*$ and $\alpha\in\R$ as 
\begin{center}$R_\ph(\alpha,Y)<\infty\quad\iff\quad\sigma(Y)<\infty.$\end{center}
Conditions~\eqref{eq:R} and \eqref{eq:phi*} become equivalent in the context of Theorem~\ref{thm1}; see the proof of Corollary~\ref{cor}(2) below. 
    \item If $\ph$ is a proper convex functional and $\alpha\in\R$, then
        \[R_\ph(\alpha,Y)\le \ph^*(Y)+\alpha,\quad Y\in\CX^*.\]
    Hence, $\dom(\ph^*)\subseteq\dom(R_\ph(\alpha,\cdot))$ for such $\alpha$, and \eqref{eq:R} implies \eqref{eq:phi*}. 
    \item In case of a $\sigma(\CX,\CX^*)$-lsc and quasiconvex functional $\ph$, $\CL_\alpha(\ph)$---and thus $\ph$ itself---can be reconstructed from $R_\ph(\alpha,\cdot)$:
\[\CL_\alpha(\ph)=\{X\in\CX\mid \E[YX]\le R_\ph(\alpha,Y),~Y\in\CX^*\}.\]
\end{enumerate}
\end{example}

Theorem~\ref{thm1} concerns a much wider class of functionals than those discussed in the preceding literature and not only unifies but genuinely improves known results. Indeed, the class of proper, $\cx$-consistent, $\sigma(\CX,\CX^*)$-lsc and quasi-star-shaped functionals extends beyond consistent risk measures and law-invariant quasiconvex functionals. Although this is already suggested by \cite[Lemma 1]{Hanetal}, we provide an explicit example.

\begin{example}\label{ex:bigger}
Consider the risk-seeking utility functions 
\[u(x)=\begin{cases}x&\text{if }x<0,\\[-0.5ex]
2x&\text{if }x\in[0,2],\\[-0.5ex]
x^2&\text{if }x>2,\end{cases}\qquad v(x)=\begin{cases}x&\text{if }x<\tfrac 1 2,\\[-0.1ex]
6x-2.5&\text{if }x\in[\tfrac 1 2,3+\sqrt{\tfrac{13}2}],\\[-0.1ex]
x^2&\text{if }x\ge 3+\sqrt{\tfrac{13}2},\end{cases}\qquad x\in\R,\]
let $C_u,C_v$ denote the associated certainty equivalent maps on the space $L^\infty$, and set
\[\ph(X):=\min\{C_u(X),C_v(X)\},\quad X\in L^\infty.\]
$\ph$ is law-invariant, continuous, and $\cx$-consistent by construction. By Lemma~\ref{lem:Svindland} below, it is also $\sigma(L^\infty,L^\infty)$-lsc. 

Functional $\ph$ is quasi-star-shaped. Indeed, for $X\in L^\infty$, $t\in\R$ and $\lambda\in(0,1)$ arbitrarily chosen, 
\begin{align*}
    \ph(\lambda X+(1-\lambda)t)&\le \min\{\max\{C_u(X),t\},\max\{C_v(X),t\}\}\\
    &=\begin{cases}t&\text{if }\min\{C_u(X),C_v(X)\}\le t\\[-0.4ex]
    \min\{C_u(X),C_v(X)\}&\text{if }\min\{C_u(X),C_v(X)\}>t
    \end{cases}\\
    &=\max\{\min\{C_u(X),C_v(X)\},t\}.
\end{align*}

Functional $\ph$ is not a consistent risk measure because it lacks cash-additivity. 
To see this, consider $X\in L^\infty$ with the property $\P(X=6)=\P(X=8)=\frac 1 2$ and compute 
\[C_u(X+c)=C_v(X+c)=\sqrt{50+14c+c^2},\quad c\ge 0.\]
\[\frac 1 2(6+c)^2+\frac 1 2(8+c)^2\]

Also, $\ph$ is not quasiconvex. More precisely,
suppose that the distribution of $X\in L^\infty$ has density 
$f=\tfrac 1 4\ind_{[-2,0]}+\frac 1 2\ind_{(0,1]}$ and define $Y$ using the sign function as $Y=\frac 1 2\sgn(X)$. In particular, $\E[Y]=0$ and $\E[X]=-\tfrac 1 4$. 
For $\lambda\in[0,1]$ we compute
\begin{align*}
\E\big[u\big(\lambda X+(1-\lambda)Y\big)\big]
  &= \E\big[\lambda X+(1-\lambda)Y\big]
   + \E\big[\big(\lambda X+\tfrac 1 2(1-\lambda)\big)\ind_{\{X>0\}}\big]\\
   &=-\tfrac{\lambda}4+\tfrac 1 2\lambda \E[X|X>0]+\tfrac 1 2\cdot\tfrac 12(1-\lambda)\\
  &= -\tfrac \lambda 4+\tfrac \lambda 4 +\tfrac 1{4}(1-\lambda)=\tfrac 1{4}(1-\lambda).
\end{align*}
Consequently, $\ph(X)\le C_u(X)=0$ and $C_u(\lambda X+(1-\lambda)Y)=\tfrac{1-\lambda}8>0$ whenever $\lambda\in[0,1)$.
Also, $\E[v(Y)]=0$, i.e., $\ph(Y)\le C_v(Y)=0$, and 
\begin{align*}
    \E[v(X)]&=\E[X]+\E[(5X-2.5)\ind_{\{X>\frac 1 2\}}]\\
    &=-\tfrac 1 4+\tfrac 1 4\E\big[(5X-2.5)\big|X>\tfrac 1 2\big]\\
    &=-\tfrac 1 4+\tfrac 5 4\cdot\tfrac 3 4-\tfrac 1 4\cdot\tfrac 5 2=\tfrac 1 {16}.
\end{align*}
As the function $[0,1]\ni \lambda\mapsto C_v(\lambda X+(1-\lambda)Y)$ is continuous, we can choose $\lambda\in(0,1)$ such that also $C_v(\lambda X+(1-\lambda)Y)>0$.  
Consequently, we establish for this $\lambda$ that 
\begin{align*}\max\{\ph(X),\ph(Y)\}&\le\max\{C_u(X),C_v(Y)\}=0\\
&<\min\big\{C_u\big(\lambda X+(1-\lambda)Y\big), C_v\big(\lambda X+(1-\lambda)Y\big)\big\}=\ph\big(\lambda X+(1-\lambda)Y\big).\end{align*}
\end{example}

In sum, Theorem~\ref{thm1} and Corollary~\ref{cor} together unify and generalise the existing results discussed in Section~\ref{sec:background}.

Regarding Theorem~\ref{thm2}, note that the quasi-star-shapedness property in \cite{Hanetal} is motivated from a preferential viewpoint, making a formulation of the collapse at the level of preferences only natural.
Also, the transition between functionals and preferences is seamless: Preference relation $\peq_\ph$ defined in \eqref{eq:peqph} satisfies \eqref{eq:starPREF} if and only if $\ph$ is quasi-star-shaped.

\cite[Proposition~1]{Hanetal}, however, proposes to capture quasi-star-shapedness axiomatically through the weaker condition
\begin{center}$X\sim t \quad\implies\quad X \peq \lambda X + (1-\lambda)t,$\end{center}
which is a relaxation of the uncertainty aversion axiom of \cite{Variational}. 
For many monotone preferences this condition is equivalent to quasi-star-shapedness and hence is subsumed by \eqref{eq:starPREF} in those cases.

\section{Proofs of the main results}\label{sec:proofs}

\subsection{Auxiliary results}

In view of our goal to make the mathematical foundations of the collapse-to-the-mean results as transparent as possible, we streamline the analysis by first isolating three auxiliary results.
The key insight is that substantial generalisation and simplification can be achieved by working directly with the convex order rather than taking a dual perspective.
The first ingredient is therefore the characterisation of the convex order found in \cite[Theorem 3.A.5]{Shaked}:

\begin{lemma}\label{lem:Shaked}
For two random variables $X,Y\in L^1$, the following are equivalent:
\begin{enumerate}[label=\tn{(\arabic*)}]
    \item $X\cx Y$.
    \item $\E[X]=\E[Y]$ and  $\operatorname{ES}_p(X) \leq \operatorname{ES}_p(Y)$ for all $p \in (0,1)$.
\end{enumerate}
\end{lemma}

Lemma~\ref{lem:Shaked} leads to the following consequence which turns out to be the main tool for the primal collapse to the mean.

\begin{lemma}\label{lem}
Let $X\in \CX$ be arbitrary and suppose $Z$ is a nonconstant simple random variable with  mean zero.
Then there are sequences $(X_n)\subseteq L^\infty$ and $(t_n)\subseteq(0,\infty)$ with the following properties:
\begin{enumerate}[label=\tn{(\alph*)}]
\item $X$ is the $\sigma(\CX,\CX^*)$-limit of $(X_n)$. 
\item $X_n\cx X$ for all $n\in\N$.
\item $X_n\cx \E[X]+tZ$ for all $t\ge t_n$ and $n\in\N$. 
\end{enumerate}
\end{lemma}

\begin{proof}
Abbreviate $c:=\E[X]$ and let $z$ be the minimal value $Z$ takes with positive probability. As $Z$ is not constant and has mean zero, $z<0$. 
Next, define for $n\in\N$
\[\Pi_n:=\big\{\{X<-n\},\{X\ge n\}\big\}\cup\big\{\{-n+(i-1)2^{-n}\le X<-n+i2^{-n}\}\mid i=1,...,2^{n+1}n\big\},\]
$\mathcal G_n:=\sigma(\Pi_n)$ to be the sub-$\sigma$-algebra of $\mathcal F$ generated by $\Pi_n$, and 
$X_n:=\E[X|\mathcal G_n]$ to be the conditional expectation of $X$ given $\mathcal G_n$. By construction, each $X_n$ is simple and satisfies $X_n\cx X$ (statement (b)).
Moreover, the sequence of $\sigma$-algebras $(\mathcal G_n)$ is increasing and generates $\sigma(X)$ in the limit. Thus, by \cite[Lemma 4.1]{General}, statement (a) holds for the sequence $(X_n)$. 

We now turn to property (c). 
Denote by $m_n$ the minimal value $X_n$ takes with positive probability, and define $\pi_n:=\min\{\P(Z=z),\P(X_n=m_n)\}>0$.
Now choose $t_n>0$ satisfying the following chain of inequalities: 
\begin{equation}\label{eq:ineqs}
    c+t_n z\le m_n\le \es_1(X_n)\le c+t_n\es_{\pi_n}(Z).
\end{equation}
The first inequality holds for $t_n$ large enough because $z<0$. The second always holds. The third is  satisfied for $t_n$ large enough because $\es_1(X_n)<\infty$ and $\es_p(Z)>\E[Z]=0$ holds for all $p\in(0,1]$.
Moreover, if $t_n$ satisfies \eqref{eq:ineqs}, then it may be replaced by any $t \ge t_n$ without violating the inequalities.

Let $p\in(0,1)$ and $t\ge t_n$.
If $\pi_n\le p<1$, we estimate 
\begin{align*}\es_p(c+tZ)-\es_p(X_n)&=c+t\es_p(Z)-\es_p(X_n)\\
&
\ge c+t\es_{\pi_n}(Z)-\es_1(X_n)\ge 0,
\end{align*}
the inequality being due to \eqref{eq:ineqs}. If $p<\pi_n$, using \eqref{eq:ineqs} again for the inequality delivers
\begin{align*}
    \es_p(c+tZ)-\es_p(X_n)&=\frac{\int_0^pq_{X_n}(s)ds-\int_0^pq_{c+tZ}(s)ds}{1-p}\\
    &=\frac{\big(m_n-(c+tz)\big)p}{1-p}\ge 0.
\end{align*}
Taking both observations together and invoking Lemma~\ref{lem:Shaked}, $X_n\cx c+tZ$. 
\end{proof}

Finally, we identify the properties of preference relations needed for the dual collapse result in Theorem~\ref{thm2}.
Loosely speaking, Lemma~\ref{lem:qss} shows that, under suitable circumstances, the infinite dimension of $\CX$ is inessential and the analysis reduces to a two-dimensional setting.

\begin{lemma}\label{lem:qss}
Suppose
$\peq$ is a preference relation as in Theorem~\ref{thm2} satisfying \eqref{eq:starPREF}. Fix $c\in\R$ and an event $A\in\CF$ with $\P(A)=\frac12$. 
Moreover, let $d\le d'$ and define \begin{center}$D := d\ind_A + d'\ind_{A^c}$.\end{center}
\begin{enumerate}[label=\tn{(\arabic*)}]
\item The set 
    \begin{equation}\label{def:U(c)}
    \mathcal U(c) := \{x\ind_A + x'\ind_{A^c}\mid x \ge x'\text{ and }c\peq x\ind_A + x'\ind_{A^c}\}
    \end{equation}
    is nonempty, and
    \[
    R_\peq(c,D) = \inf_{X\in\mathcal U(c)} \E[DX].
    \]
\item If $(X_n)\subseteq\mathcal U(c)$ is unbounded in norm and $s_n:=\|X_n\|_\infty^{-1}$, then there exists a subsequence $(n_k)$ such that $U:=\lim_{k\to\infty}s_{n_k}X_{n_k}$ satisfies  $\|U\|_\infty=1$ and 
    \[
    c\peq x + t U\quad \text{whenever }c\peq x.
    \]
\end{enumerate}
\end{lemma}
\begin{proof}
\begin{enumerate}[label=(\arabic*)]
    \item 
    Suppose $c\peq X\in\CX$ and let $x,x'\in\R$ be such that $\E[X|\sigma(A)]=x\ind_A+x'\ind_{A^c}$. 
    As $\peq$ is law invariant and $\P(A)=\P(A^c)=\frac 12$, we have \begin{center}$x'\ind_A+x\ind_{A^c}\sim x\ind_A+x'\ind_{A^c}$.\end{center}
    As $\peq$ is $\cx$-consistent,  
    \[x'\ind_A+x\ind_{A^c}\sim x\ind_A+x'\ind_{A^c}\succeq X.\]
    By transitivity of $\peq$, both $x'\ind_A+x\ind_{A^c}$ and $x\ind_A+x'\ind_{A^c}$ are elements of $\mathcal U(c)$. 
    Moreover, $D$ being $\sigma(A)$-measurable implies 
    \begin{align*}R_\peq(c,D)&=\inf\{\E\big[D\E[X|\sigma(A)]\big]\mid c\peq X\}\ge\inf_{X\in\mathcal U(c)}\E[DX]\ge R_\peq(c,D).\end{align*}
    The first inequality is due to the Hardy-Littlewood inequality \cite[Theorem A.28]{FoeSch}.  
    \item 
    The set 
    $\mathcal U(c)$ is homeomorphic to a closed subset of $\R^2$. 
    Hence, there exists a subsequence $(n_k)$ and $\sigma(A)$-measurable $U\in L^\infty$ with 
    $\Norm_\infty$-norm 1 such that $s_{n_k}X_{n_k}\to U$ in $L^\infty$ as $k\to\infty$.
    Let $x\in\R$ be such that $c\peq x$. 
    By \eqref{eq:starPREF}, we have for all $t>0$ and $k$ large enough that 
    \[ts_{n_k}X_{n_k}+(1-ts_{n_k})x\succeq c.\]
    Using the upper semicontinuity property of $\peq$, we obtain $x+tU\succeq c$. 
\end{enumerate}
\end{proof}

We shall now prove the most general result Theorem~\ref{thm2} first and derive Theorem~\ref{thm1} and Corollary~\ref{cor} as applications thereof.

\subsection{Proof of Theorem~\ref{thm2}}

\begin{enumerate}[label=(\arabic*)]
    \item The implications $\ref{P1}\implies\ref{P2}\implies\ref{P4}\implies\ref{P5}$ and $\ref{P1}\implies\ref{P3}\implies\ref{P5}$ are clear. 

$\ref{P5}\implies\ref{P1}$:
Select $c\in\R$ and nonconstant $Z\in\CX$ with $\E[Z]=0$ such that
$c+tZ\sim c$ holds for all $t\ge 0$.
In particular, as every sub-$\sigma$-algebra $\CG\subseteq\CF$ satisfies $c\cx c+t\E[Z|\CG]\cx c+tZ$, we have 
\[c\sim c+tZ\peq c+t\E[Z|\CG]\peq c,\]
i.e., we may assume without loss of generality that $Z$ is simple. 

Let $X\in\CX$ with $\E[X]=c$, which means that $X\peq c$. 
Consider the sequences $(X_n)\subseteq L^\infty$ and $(t_n)\subseteq (0,\infty)$ constructed in Lemma~\ref{lem}. 
For all $n\in\N$, 
\[X_n\succeq c+t_nZ\sim c.\]
Using that $\peq$ is $\sigma(\CX,\CX^*)$-usc, we obtain $c\peq X$ by taking the limit. 
Hence, $X\sim c=\E[X]$.
As $c$ was chosen arbitrarily, this shows that
\[X\peq Y\quad \iff\quad\E[X]\peq \E[Y].\]

$\ref{P1}\implies\ref{P6}$:  Let $Y\in\CX^*\setminus\R$ and fix $Z\in\CX$ with $\E[Y]\E[Z]\neq\E[YZ]$. 
In view of assumption \ref{P1}, we have for all $X\in\CX$ and $t\in\R$ that $X+t(Z-\E[Z])\sim X$.
Consequently,
\[R_\peq(X,Y)\le\inf_{t\in\R}\E\big[(X+t(Z-\E[Z]))Y\big]=\E[XY]+\inf_{t<0}t|\E[ZY]-\E[Y]\E[Z]|=-\infty.\]
\item We will show the implication $\ref{P6}\implies\ref{P5}$ in several steps.
Let $c\in\R$, $a\in\{\pm1\}$, and fix an event $A\in\CF$ with $\P(A)=\frac 12$.
For all $n\in\N$,
\[Y_n^{(a)}:=a-\tfrac 1 n\ind_A+\tfrac 1 n\ind_{A^c}\] 
satisfies 
\[R_\peq(c,Y_n^{(a)})=\inf\{\E[Y_n^{(a)}Z]\mid Z\succeq c\}=-\infty\] 
by assumption. 
Invoking Lemma~\ref{lem:qss}(1), we can pick $X_{n}^{(a)}$ in the set $\mathcal U(c)$ defined by \eqref{def:U(c)} such that $\E[Y_n^{(a)}X_n^{(a)}]\le-n$.
The sequence $(X_n^{(a)})$ cannot be bounded. 
Let $U^{(a)}$ be the limit of the rescaled sequence constructed in Lemma~\ref{lem:qss}(2). 
    By construction, \[\E[aU^{(a)}]=\lim_{k\to\infty}\frac{\E[Y_{n_k}^{(a)}X_{n_k}^{(a)}]}{\|X_{n_k}^{(a)}\|_\infty}\le \limsup_{k\to\infty}\frac{-n_k}{\|X_{n_k}^{(a)}\|_\infty}\le 0,\]
or equivalently, 
\begin{equation}\label{eq:U1U-1}\E[U^{(1)}]\le 0\quad\text{and}\quad\E[U^{(-1)}]\ge 0.\end{equation}
 
 \textsc{Case 1:} One of the inequalities in \eqref{eq:U1U-1} is an equality.  
Choosing $U$ appropriately in $\{U^{(1)},U^{(-1)}\}$, Lemma~\ref{lem:qss}(2) delivers 
\begin{equation}\label{eq:concl1}
    c\peq c+tU,\quad t>0.
\end{equation}

\textsc{Case 2:} Both inequalities in \eqref{eq:U1U-1} are strict.  
As $\peq$ is $\cx$-consistent, we have 
\begin{align*}c\peq c+tU^{(1)}\peq c+t\E[U^{(1)}],\quad t\ge 0,\end{align*}
which means $c\peq x$ for all $x\le c$.
Likewise, 
\begin{align*}c\peq c+tU^{(-1)}\peq c+t\E[U^{(-1)}],\quad t\ge 0,\end{align*}
i.e., $c\peq x$ for all $x\ge c$. 

Fix $X\in\mathcal U(c)$ and note that we have just shown that also $c\peq -\E[X]$. 
By \eqref{eq:starPREF},
    $c\peq \tfrac{X-\E[X]}2$, i.e., 
    \[\mathcal K_0:=\Big\{\tfrac{X-\E[X]}2\,\Big|\,X\in\mathcal U(c)\Big\}\subseteq\mathcal U(c).\]
    Consider the random variable $Y=-\ind_A+\ind_{A^c}$ with $\E[Y]=0$. By \ref{P6}, $R_\peq(c,Y)=-\infty$, leading to
    \[-\infty=R_\peq(c,Y)=\inf_{X\in\mathcal U(c)}\E[YX]=2\inf_{X\in\mathcal U(c)}\E\big[Y\tfrac{X-\E[X]}2\big]= 2\inf_{X\in\mathcal K_0}\E[YX].\]
    Hence, the set $\mathcal K_0$ cannot be norm-bounded. 
    By Lemma~\ref{lem:qss}(2), there exists a $\sigma(A)$-measurable $U$ satisfying $\E[U]=0$, $\|U\|_\infty=1$, and \begin{equation}\label{eq:concl2}c\peq c+tU,\quad t>0.\end{equation}
    
    Summing up \eqref{eq:concl1} and \eqref{eq:concl2}, we have shown statement \ref{P5}.
\end{enumerate}\hfill\qed

\subsection{Proof of Theorem~\ref{thm1}}

\begin{enumerate}[label=(\arabic*)]
\item Define the preference relation $\peq_\ph$ by \eqref{eq:peqph}.
Its $\cx$-consistency and upper semicontinuity follow directly from $\cx$-consistency and lower semicontinuity of $\ph$.
Moreover, for any $c\in\R$, $t\ge 0$, and $Z\in\CX$ with $\E[Z]=0$, we have $c\cx c+tZ$. Hence,  $\ph(c+tZ)\le\ph(c)$ implies $\ph(c+tZ)=\ph(c)$, or equivalently, $c\sim_\ph c+tZ$.
Consequently, the equivalences of \ref{F1}--\ref{F5} in Theorem~\ref{thm1} and the implication $\ref{F1}\implies
\ref{F7}$ are a mere reformulation of Theorem~\ref{thm2} applied to $\peq_\ph$.

$\ref{F1}\implies\ref{F6}$: Select any nonconstant $Z\in\CX$ with mean zero and $a=0$. 

$\ref{F6}\implies\ref{F5}$: Let $c,t\in\R$. 
Using $\cx$-consistency for the inequality delivers
\[\ph\big(c+t(Z-\E[Z])\big)=\ph\big(c-t\E[Z]\big)+at\le\ph(c-tZ)+at=\ph(c).\]

\item As shown in (1), $\ph$ is expectation-based. 
By virtue of properness, we can select $X_0\in\dom(\ph)$; by $\cx$-consistency of $\ph$, $c_0:=\E[X_0]\in\dom(\ph)\cap\R$.
Let $X\in\CX$ and compute 
\[\ph(c_0)=\ph\big(\E[X]+\tfrac{c_0-\E[X]}{\E[Z]}\E[Z]\big)=\ph\big(X+\tfrac{c_0-\E[X]}{\E[Z]}Z\big)=\ph(X)+\tfrac{(c_0-\E[X])a}{\E[Z]}.\]
Rearrange this identity. 

\item 
Assume that $\ph$ is quasi-star-shaped and let $X\in\CX$ and $c,t\in\R$ satisfy $c\peq_\ph X$ and $c\peq_\ph t$, or equivalently, $\ph(c)\ge\max\{\ph(X),\ph(t)\}$.
Quasi-star-shapedness then yields for arbitrary $\lambda\in[0,1]$ that
$c\peq_\ph\lambda X+(1-\lambda)t$.
Hence, the implication $\ref{F7}\implies\ref{F1}$ follows from the implication $\ref{P6}\implies\ref{P1}$ in Theorem~\ref{thm2} and Lemma~\ref{lem:relation}.
\item Suppose \ref{F7} holds and that $\ph$ is proper, convex, $\sigma(\CX,\CX^*)$-lsc, and law invariant.
Since $\dom(\ph)\cap\R\neq\varnothing$, Example~\ref{ex}(2) implies $\dom(\ph^*)\subseteq\R$.
Conversely, if $\dom(\ph^*)\subseteq\R$, then $\ph$ is expectation-based by the Fenchel–Moreau representation.
\end{enumerate}\hfill\qed

\subsection{Proof of Corollary~
\ref{cor}}

First, $\ph$ being proper implies the existence of $X\in\CX$ with $\ph(X)<\infty$. 
As $\E[X]\cx X$, $\cx$-consistency implies
\[\ph(0)=\ph(\E[X])-\E[X]\le \ph(X)-\E[X]<\infty.\]

\begin{enumerate}[label=(\arabic*)]
    \item $\ref{C1}\implies\ref{C4}$ and $\ref{C1}\implies\ref{C3}$: These implications are clear.

    $\ref{C3}\implies\ref{C2}$: Let $Z$ be as in statement \ref{C3} and $t\ge 0$. Using cash-additivity for the first equality and $\cx$-consistency for the estimate, we verify that 
    \[\ph(t\E[Z]-tZ)=t\E[Z]-at+\ph(0)=\ph(t\E[Z])-at\le \ph(tZ)-at=\ph(0).\]
    Hence, the nonconstant zero-mean random variable $U:=\E[Z]-Z$ satisfies $\sup_{t\ge 0}\ph(tU)\le \ph(0)$. 

    $\ref{C2}\implies\ref{C1}$: 
    Let $Z$ be as described by statement \ref{C2}. 
    For all $c\in\R$ and $t\ge 0$, 
    \[\ph(c+tZ)=c+\ph(tZ)\le c+\ph(0)=\ph(c).\]
    Apply Theorem~\ref{thm1} to obtain 
    \[\ph(X)=\ph(\E[X])=\ph(0)+\E[X],\quad X\in\CX.\] 

    \item Suppose that $\ph(0)=0$, $\ph$ is star-shaped, and $\dom(\ph^*)\subseteq\R$. 
    Moreover, let $\alpha\in\ph(\R)=\R$ and $Y\in\CX^*$ be such that $R_\ph(\alpha,Y)<\infty$.
    By Example~\ref{ex}(1), $\sigma(Y)<\infty$. 
    On the other hand, the inclusion $\{\E[X]\mid X\in\CA_\ph\}\subseteq\CA_\ph$ stemming from $\cx$-consistency implies that also  $\sigma(1)<\infty$. 
    Set 
    \[\widetilde Y:=\frac{Y+\E[|Y|]+1}{\E[Y]+\E[|Y|]+1}.\]
    As $\sigma$ is positively homogeneous and subadditive, we have 
    \[\ph^*(\widetilde Y)=\sigma(\widetilde Y)<\infty\]
    by Example~\ref{ex}(1), meaning that $\widetilde Y$ and {\em a fortiori} $Y$ itself are constant. In other words, \eqref{eq:R} holds. 
    
    Moreover, a cash-additive $\ph$ is quasi-star-shaped if it is star-shaped; see the proof of  \cite[Proposition 2(ii)]{Hanetal}. The implication (C4)$\implies$(C1) thus follows with Theorem~\ref{thm1}. 
\end{enumerate}
   \hfill\qed

\subsection{A comment on the topological assumptions}\label{sec:topology}

Throughout the manuscript, we work with topological assumptions involving the locally convex topology $\sigma(\CX,\CX^*)$ introduced in Section~\ref{sec:setting}. 
Many readers might be more comfortable with norm topologies instead, and given that Theorem~\ref{thm1}, for example, concerns not only quasiconvex functionals, their lower semicontinuity properties with respect to norm and weak topologies differ {\em a priori}. 
Another limitation is that checking lower semicontinuity with respect to $\sigma(\CX,\CX^*)$ might be nontrivial.
This closing subsection explains where the topology is needed and sketches potential alternatives. 

The topology $\sigma(\CX,\CX^*)$ plays a key role in establishing Lemma \ref{lem}, i.e., in showing that an arbitrary $X\in\mathcal X$ can be approximated by its conditional expectations along an increasing sequence of finite partitions. 
Conditional expectations are important for compatibility with the convex order. 
As shown by \cite[Lemma 4.1]{General}, this approximation works seamlessly in the $\sigma(\mathcal X,\mathcal X^*)$-topology.
It is also possible in order-continuous norm topologies, for example those of Orlicz spaces satisfying the $\Delta_2$-condition or classical $L^p$-spaces with $1\le p<\infty$; see \cite[p.\ 401]{Orlicz}. 
In general normed function spaces $\CX$, however, the feasibility of such an approximation implies that $L^\infty$ is a dense subspace of $\CX$. 
Hence, this approximation is impossible in the norm topology of many spaces, including all Orlicz spaces without the $\Delta_2$-property. 

There is, however, a feasible alternative that does not require fixing a dual space $\CX^*$ {\em a priori}. 
We say that a sequence $(X_n)\subseteq\CX$ $o$-converges to $X$ if $X_n\to X$ a.s.\ and there exist $Y_1,...,Y_K\in\CX$ and $\alpha_1,...,\alpha_K>0$ such that $\sup_{n\in\N}|X_n|\le\sum_{i=1}^K\alpha_i|Y_i|$.\footnote{~While this connection is not relevant here, $o$-convergence is sequential order convergence in the sub-ideal of $L^1$ generated by $\CX$.} 
A proper functional $\ph\colon\CX\to(-\infty,\infty]$ has the {\em Fatou property} if
\[\ph(X)\le\liminf_{n\to\infty}\ph(X_n)\]
for all $(X_n)$ $o$-convergent to $X$. Moreover, we denote the $L^1$-norm by $\|\cdot\|_1$.

\begin{lemma}\label{lem:Fatou}
    Let $(\CX,\CX^*)$ be a pair of spaces as described in Section~\ref{sec:setting} and let $\ph\colon\CX\to(-\infty,\infty]$ be proper and $\cx$-consistent. The following statements are equivalent:
    \begin{enumerate}[label=\tn{(\alph*)}]
        \item $\ph$ is $\sigma(\CX,\CX^*)$-lsc.
        \item $\ph$ has the Fatou property.
        \item $\ph$ is lsc in the relative $L^1$-topology on $\CX$. 
    \end{enumerate} 
\end{lemma}
\begin{proof}
    A $\cx$-consistent functional has the property of dilatation monotonicity; see \cite{Rahsepar}. 
    By \cite[Proposition 3]{Rahsepar}, (b) and (c) are both equivalent to $\sigma(\CX,\mathcal S)$-lower semicontinuity, where $\mathcal S$ denotes the space of all simple random variables. 
    As $\sigma(\CX,\mathcal S)$ is weaker than $\sigma(\CX,\CX^*)$, we also obtain (b) $\implies$ (a). The implication (a) $\implies$ (c) is due to every $\Norm_1$-convergent sequence $(X_n)\subseteq\CX$ with limit in $\CX$ being $\sigma(\CX,\CX^*)$-convergent.  
\end{proof}

Lemma~\ref{lem:Fatou} transfers to preference relations by applying it to functionals 
\[\ph_X(Y):=\begin{cases}
    0&\text{if }X\peq Y,\\[-0.5ex]
    \infty&\text{otherwise},
\end{cases}\qquad Y\in\CX.\]
Hence, the assumption of $\sigma(\CX,\CX^*)$-lower (or upper) semicontinuity can also be replaced by appropriate versions of (b) and (c) in all preceding results. 

Lemma~\ref{lem:Fatou} seemingly does not cover one of the most prominent cases,  $\CX=L^\infty$. 
However, arguing as in Step~2 of the proof of \cite[Proposition~1.2]{Svindland}:
\begin{lemma}\label{lem:Svindland}
    Every proper, norm-lsc, and $\cx$-consistent functional $\ph$ on $L^\infty$ is $\sigma(L^\infty,\mathcal S)$-lsc. 
\end{lemma}
Alternatively, one can avoid $\sigma(L^\infty,\CX^*)$ topologies altogether in this case and pair $L^\infty$ with its norm dual, which is commonly identified with the space $\mbf{ba}$ of finitely additive and bounded set functions $\mu\colon L^\infty\to\R$ (the pairing being provided by Dunford-Schwartz integrals).
The proofs then go through \emph{verbatim}.

   \medskip

{\bf Competing interests:} The author declares none.

\medskip

{\bf Acknowledgements:} The author is indebted to two anonymous referees for their detailed comments, and to Benjamin C\^ot\'e and Christian Laudag\'e for their feedback on earlier drafts of the manuscript.


\medskip

\begin{thebibliography}{xx}

{\footnotesize

\harvarditem{Amarante}{2021}{Bipolar}
Amarante, M. \harvardyearleft 2021\harvardyearright, Bipolar behavior of submodular, law-invariant capacities.
\textit{Statistics \& Risk Modeling}, {\bf 38}(3--4), 65--70. 

\harvarditem{Amarante et al.}{2024}{Amarante}
Amarante, M., Liebrich, F.-B. \harvardand\ Munari, C. \harvardyearleft 2024\harvardyearright, Uniqueness of convex-ranged probabilities and applications to risk measures and games. \textit{Mathematics of Operations Research}, \textbf{50}(1), 743--763.

\harvarditem{Artzner et al.}{2009}{Artzner}
Artzner, P., Delbaen, F. \harvardand\ Koch-Medina, P. \harvardyearleft 2009\harvardyearright,
Risk measures and efficient use of capital.
{\em ASTIN Bulletin}, \textbf{39}(1), 101--116.

\harvarditem{Bellini et al.}{2021a}{General}
Bellini, F., Koch-Medina, P., Munari, C. \harvardand\ Svindland, G. \harvardyearleft 2021a\harvardyearright, Law-invariant functionals on general spaces of random variables.
\textit{SIAM Journal on Financial Mathematics}, \textbf{12}(1), 318--341.

\harvarditem{Bellini et al.}{2021b}{Bellini}
Bellini, F., Koch-Medina, P., Munari, C. \harvardand\ Svindland, G. \harvardyearleft 2021b\harvardyearright, Law-invariant functionals that collapse to the mean. \textit{Insurance: Mathematics and Economics}, \textbf{98}, 83--91.

\harvarditem{Castagnoli et al.}{2004}{Castagnoli}
Castagnoli, E., Maccheroni, F., 
\harvardand\ Marinacci, M. \harvardyearleft 2004\harvardyearright, Choquet insurance pricing: A caveat. {\em Mathematical Finance}, {\bf 14}(3), 481--485.

\harvarditem{Centrone and Rosazza Gianin}{2025}{Centrone}
Centrone, F. \harvardand\ Rosazza Gianin, E. \harvardyearleft 2025\harvardyearright, Capital allocation rules and generalized collapse to the mean: Theory and practice. \textit{Mathematics}, \textbf{13}(6), 964.


\harvarditem{Chen et al.}{2021}{Chen}
Chen, S., Gao, N., Leung, D.\ H., \harvardand\ Li, L.\ \harvardyearleft 2021\harvardyearright, Do law-invariant linear functionals collapse to the mean? Preprint, \url{arXiv:2107.11239v2}.  

\harvarditem{Chen et al.}{2022}{Automatic}Chen, S., Gao, N., Leung, D.\ H. \harvardand\ Li, L. \harvardyearleft 2022\harvardyearright, Automatic Fatou property of law-invariant risk measures. {\em Insurance: Mathematics and Economics} {\bf 105}, 41--53.

\harvarditem{Cheung et al.}{2014}{Cheung}
Cheung, K.\ C., Dhaene, J., Lo, A., \harvardand\ Tang, Q. \harvardyearleft 2014\harvardyearright, Reducing risk by merging counter-monotonic risks. 
{\em Insurance: Mathematics and Economics},
{\bf 54}, 58--65.

\harvarditem{C\^ot\'e and Wang}{2026}{Cote}Côté, B. \harvardand\ Wang, R. \harvardyearleft 2026\harvardyearright, 
On convex order and supermodular order without finite mean. {\em Insurance: Mathematics and Economics}, \url{https://doi.org/10.1016/j.insmatheco.2026.103234}.

\harvarditem{Föllmer \harvardand\ Schied}{2016}{FoeSch}Föllmer, H. \harvardand\ Schied, A.  \harvardyearleft 2016\harvardyearright , {\em Stochastic Finance: An Introduction in Discrete Time}, De Gruyter.

\harvarditem{Frittelli \harvardand\ Rosazza Gianin}{2005}{Frittelli}
Frittelli, M. \harvardand\ Rosazza Gianin, E. \harvardyearleft 2005\harvardyearright, Law invariant convex risk measures. \textit{Advances in Mathematical Economics}, \textbf{7}, 33--46.


\harvarditem{Gao et al.}{2018}{Orlicz}Gao, N., Leung, D., Munari, C. \harvardand\ Xanthos, F.\ \harvardyearleft 2018\harvardyearright, Fatou property, representations, and extensions of law-invariant risk measures on general Orlicz spaces. \textit{Finance and Stochastics} {\bf 22}, 395--415.



\harvarditem{Gilboa}{2009}{Gilboa}
Gilboa, I. \harvardyearleft 2009\harvardyearright, \textit{Theory of Decision under Uncertainty}, Cambridge University Press.


\harvarditem[Han et~al.]{Han, Wang, Wang \harvardand\ Xia}{2025}{Hanetal}
Han, X., Wang, Q., Wang, R. \harvardand\ Xia, J.  \harvardyearleft 2025\harvardyearright, Cash-subadditive risk measures without quasi-convexity. {\em Mathematics of Operations Research}, \url{https://doi.org/10.1287/moor.2022.0312}.

\harvarditem{He et al.}{2016}{He}
He, J., Tang, Q., \harvardand\ Zhang, H. \harvardyearleft 2016\harvardyearright, 
Risk reducers in convex order. 
{\em Insurance: Mathematics and Economics} {\bf 70}, 80--88.

\harvarditem{Herdegen et al.}{2024}{herdegen2024}
Herdegen, M., Khan, N. \harvardand\  Munari, C. \harvardyearleft 2024\harvardyearright, Risk, utility and sensitivity to large losses. Preprint,  \url{arXiv:2405.12154v1}.

\harvarditem{Herdegen et al.}{2025}{herdegen2025}
Herdegen, M., Khan, N. \harvardand\  Munari, C. \harvardyearleft 2025\harvardyearright, How to reduce risk by increasing risk. Preprint, \url{https://dx.doi.org/10.2139/ssrn.5083558}.

\harvarditem{Kaas et al.}{2008}{Kaas}Kaas, R., Goovaerts, M., Dhaene, J. \harvardand\ Denuit, M. \harvardyearleft 2008\harvardyearright, {\em Modern Actuarial Risk Theory}, Springer.

\harvarditem{Liebrich}{2024}{Liebrich2}
Liebrich, F.-B. \harvardyearleft 2024\harvardyearright, Risk sharing under heterogeneous beliefs without convexity. \textit{Finance and Stochastics}, \textbf{28},~999--1033.

\harvarditem{Liebrich \harvardand\ Munari}{2022}{Liebrich}
Liebrich, F.-B. \harvardand\ Munari, C.  \harvardyearleft 2022\harvardyearright, Law-invariant functionals that collapse to the mean: Beyond convexity. {\em Mathematics and Financial Economics}, {\bf 16}(3),~447--480.

\harvarditem{Liebrich and Svindland}{2019}{Liebrich1}Liebrich, F.-B. \harvardand\ Svindland, G. \harvardyearleft 2019\harvardyearright, Risk sharing for capital requirements with multidimensional security markets.
{\em Finance and Stochastics} {\bf 23}, 925--973.

\harvarditem{Maccheroni et al.}{2006}{Variational}Maccheroni, F., Marinacci, M., \harvardand\ Rustichini, A. \harvardyearleft 2006\harvardyearright, Ambiguity aversion, robustness, and the variational representation of preferences.
\textit{Econometrica}, \textbf{74}(6), 1447--1498.

\harvarditem{Maccheroni et al.}{2025}{Maccheroni}
Maccheroni, F., Marinacci, M., Wang, R. \harvardand\ Wu, Q. \harvardyearleft 2025\harvardyearright, Risk aversion and insurance propensity. \textit{American Economic Review}, \textbf{115}(5), 1597--1649.

\harvarditem{Mao \harvardand\ Wang}{2020}{Consistent}
Mao, T. \harvardand\ Wang, R.  \harvardyearleft 2020\harvardyearright , Risk aversion in regulatory capital principles. {\em SIAM Journal on Financial Mathematics}, {\bf 11}(1), 169--200.

\harvarditem{Markowitz}{1952}{Markowitz}Markowitz, H. \harvardyearleft 1952\harvardyearright, Portfolio Selection. {\em The Journal of Finance} {\bf 7}(1), 77--91.



\harvarditem{McNeil et al.}{2015}{Embrechts}
McNeil, A.\ J., Frey, R., \harvardand\ Embrechts, P. \harvardyearleft 2015\harvardyearright, {\em Quantitative Risk Management: Concepts,
Techniques and Tools}, revised edition, Princeton University Press.



\harvarditem{Rabin}{2000}{Rabin}Rabin, M. \harvardyearleft 2000\harvardyearright, Risk aversion and expected-utility theory: a calibration theorem. {\em Econometrica}, {\bf 68}, 1281--1292.

\harvarditem{Rahsepar and Xanthos}{2020}{Rahsepar}Rahsepar, M. \harvardand\ Xanthos, F. \harvardyearleft 2020\harvardyearright, On the extension property of dilatation monotone risk measures. {\em Statistics \& Risk Modeling} {\bf 37}(3-4), pp.\ 107--119.

\harvarditem{Shaked \harvardand\ Shanthikumar}{2007}{Shaked}
Shaked, M. \harvardand\ Shanthikumar, J. G.  \harvardyearleft 2007\harvardyearright , {\em Stochastic Orders}, Springer.

\harvarditem{Svindland}{2010}{Svindland}Svindland, G.  \harvardyearleft 2010\harvardyearright, Continuity properties of law-invariant (quasi-)convex risk functions on $L^\infty$. {\em Mathematics and Financial Economics}, {\bf 3}, 39--43.

\harvarditem{Wang}{2000}{Wang2000} Wang, S.\ S.\ \harvardyearleft 2000\harvardyearright, A class of distortion operators for pricing financial and insurance
risks. {\em Journal of Risk and Insurance}, {\bf 67}(1), 15--36.

\harvarditem{Wang}{2002}{Wang2002}
Wang, S.\ S. \harvardyearleft 2002\harvardyearright, A universal framework for pricing financial and insurance risks. {\em Astin Bulletin}, {\bf 32}(2), 213--234.

\harvarditem{Wang et al.}{1997}{Wangetal}Wang, S.\ S., Young, V.\ R., Panjer, H.\ H. \harvardyearleft 1997\harvardyearright, Axiomatic characterization of insurance
prices. {\em Insurance: Mathematics and Economics}, {\bf 21}, 173--183.
}


\end{thebibliography}
 \end{document}